\documentclass{article}
\usepackage[authoryear,sort]{natbib}
\usepackage{fullpage}

\usepackage{comment}

\usepackage[utf8]{inputenc} 
\usepackage[T1]{fontenc}    
\usepackage{hyperref}       
\usepackage{url}            
\usepackage{booktabs}       
\usepackage{amsfonts}       
\usepackage{nicefrac}       
\usepackage{microtype}      

\usepackage{amsmath,amssymb,mathtools}
\usepackage{bm}
\usepackage{amsthm}
\usepackage{graphicx,subfig,color,wrapfig,epsfig,epstopdf}
\usepackage{enumitem}
\usepackage{booktabs}
\usepackage{algorithm,algorithmic}

\def\x{\boldsymbol{x}}
\def\g{\boldsymbol{g}}
\def\v{\boldsymbol{v}}
\def\bdelta{\boldsymbol{\delta}}

\def\bG{\overline{G}}

\def\of{\overline{f}}

\def\eff{\mathrm{eff}}

\def\OA{\texttt{DeepSqueeze}}

\newcommand\numberthis{\addtocounter{equation}{1}\tag{\theequation}}

\allowdisplaybreaks[4]

\newtheorem{theorem}{Theorem}
\newtheorem{lemma}{Lemma}
\newtheorem{corollary}{Corollary}
\newtheorem{assumption}{Assumption}

\title{$\OA$: Decentralization Meets Error-Compensated Compression}

\usepackage{authblk}

\author[1]{Hanlin Tang}
\author[1]{Xiangru Lian}
\author[5]{Shuang Qiu}
\author[3]{Lei Yuan}
\author[4]{Ce Zhang}
\author[2]{Tong Zhang}
\author[3,1]{Ji Liu}
\affil[1]{Department of Computer Science, University of Rochester}
\affil[2]{Hong Kong University of Science and Technology}
\affil[3]{Seattle AI Lab, FeDA Lab, Kwai Inc}
\affil[4]{Department of Computer Science, ETH Zurich}
\affil[5]{University of Michigan}

\begin{document}

\maketitle

\begin{abstract}
Communication is a key bottleneck in distributed training. Recently, an \emph{error-compensated} compression technology was particularly designed for the \emph{centralized} learning and receives huge successes, by showing significant advantages over state-of-the-art compression based methods in saving the communication cost. Since the \emph{decentralized} training has been witnessed to be superior to the traditional \emph{centralized} training in the communication restricted scenario, therefore a natural question to ask is ``how to apply the error-compensated technology to the decentralized learning to further reduce the communication cost.'' However, a trivial extension of compression based centralized training algorithms does not exist for the decentralized scenario. key difference between centralized and decentralized training makes this extension extremely non-trivial. In this paper, we propose an elegant algorithmic design to employ error-compensated stochastic gradient descent for the decentralized scenario, named $\OA$. Both the theoretical analysis and the empirical study are provided to show the proposed $\OA$ algorithm outperforms the existing compression based decentralized learning algorithms. To the best of our knowledge, this is the first time to apply the error-compensated compression to the decentralized learning.
\end{abstract}

\section{Introduction}
We consider the following decentralized optimization:
\begin{equation}
\min_{\bm{x}\in\mathbb{R}^{N}}\quad f(\bm{x}) = {1\over n} \sum_{i=1}^n \underbrace{\mathbb{E}_{\xi\sim\mathcal{D}_i}F_{i}(\bm{x}; \bm{\xi})}_{=: f_i(\bm{x})},\label{eq:main}
\end{equation}
where $n$ is the number of node and $\mathcal{D}_i$ is the local data distribution for node $i$. $n$ nodes form a connected graph and each node can only communicate with its neighbors. 

Communication is a key bottleneck in distributed training \citep{seide2016cntk,Abadi:2016:TSL:3026877.3026899}. A popular strategy to reduce communication cost is compressing the gradients computed on the local workers and sending to the parameter server or a central node. To ensure the convergence, the compression usually needs to be unbiased and the compression ratio needs to be chosen very cautiously (it cannot be too aggressive), since the bias or noise caused by the compression may significantly degrade the convergence efficiency. Recently, an error-compensated technology has been designed to make the aggressive compression possible. The key idea is to store the compression error in the previous step and send the compressed sum of the gradient and the remaining compression error in the previous step \citep{tang_ds}:
\begin{align*}
\bm{g}'\gets & C_{\omega}\left[\bm{g} + \bm{\delta}\right]\quad && \text{(compressed gradient on local worker(s))}\\
\bm{\delta}\gets & (\bm{g} + \bm{\delta}) - C_{\omega}\left[\bm{g} + \bm{\delta}\right] && \text{(remaining error on the local worker(s))}\\
\bm{x} \gets & \bm{x} - \gamma\bm{g}' && \text{(update on the parameter server)}
\end{align*}
where $\bm{g}$ is the computed stochastic gradient. This idea has been proven to very effective and may significantly outperform the traditional (non-error-compensated) compression centralized training algorithms such as \citet{pmlr-v70-zhang17e,NIPS2018_7519}.

The \emph{decentralized} training has been proven to be superior to the \emph{centralized} training in terms of reducing the communication cost \citep{boyd2006randomized}, especially when the model is huge and the network bandwidth and latency are less satisfactory. The key reason is that in the decentralized learning framework workers only need to communicate with their individual neighbors, while in the centralized learning all workers are required to talk to the central node (or the parameter server). Moreover, the iteration complexity (or the total computational complexity) of the decentralized learning is proven to be comparable to that of the centralized counterpart \citep{lian2017can}.

Therefore, given the recent successes of the error-compensated technology for the centralized learning, it motivates us to ask a natural question: \emph{How to apply the error-compensated technology to the decentralized learning to further reduce the communication cost?} 

However, key differences exist between centralized and decentralized training, and it is highly non-trivial to extend the error-compensated technology to decentralized learning. In this paper, we propose an error-compensated decentralized stochastic gradient method named $\OA$ by combining these two strategies. Both theoretical analysis and empirical study are provided to show the advantage of the proposed $\OA$ algorithm over the existing compression based decentralized learning algorithms, including \citet{NIPS2018_7992, Koloskova:2019aa}. 

\paragraph{Our Contribution:}
\begin{itemize}
\item To the best of our knowledge, this is the first time to apply the error-compensated compression technology to decentralized learning.
\item Our algorithm is compatible with almost all compression strategies, and admits a much higher tolerance on aggressive compression ratio than existing work (for example,  \citep{NIPS2018_7992}) in both theory and empirical study.  
\end{itemize}

\paragraph{Notations and definitions}
Throughout this paper, we use the following notations:
\begin{itemize}
\item $\nabla f(\cdot)$ denotes the gradient of a function $f$.
\item $f^{*}$ denotes the optimal value of the minimization problem \eqref{eq:main}.
\item $f_i(x) := \mathbb{E}_{\xi\sim\mathcal{D}_i}F_{i}(x; \xi)$.
\item $\lambda_{i}(\cdot)$ denotes the $i$-th largest eigenvalue of a symmetric matrix.
\item $\bm{1}=[1,1,\cdots,1]^{\top}\in\mathbb{R}^n$ denotes the full-one vector.
\item $A_n = \frac{\bm{1}\bm{1}^{\top}}{n} \in\mathbb{R}^{n\times n} $ is an all $\frac{1}{n}$ square matrix.
\item $\|\cdot\|$ denotes the $\ell_2$ norm for vectors and the spectral norm for matrices.
\item $\|\cdot\|_F$ denotes the vector Frobenius norm of matrices.
\item $\bm{C}_{\omega}(\cdot)$ denotes the compressing operator.
\end{itemize}

The rest of the paper is organized as follows. First we review the related studies. We then discuss our proposed method in detail and provides key theoretical results. We further validate our method in an experiment and finally we conclude this paper.


\section{Related Work}
In this section, we review the recent works in a few areas that are closely related to the topic of this study: distributed learning, decentralized learning, compression based learning, and error-compensated compression methods.

\subsection{Centralized Parallel Learning}

Distributed learning is an widely-used acceleration strategy for training deep neural network with high computational cost. Two main designs are developed to parallelize the computation: 1) centralized parallel learning \citep{agarwal2011distributed, recht2011hogwild}, where all the workers are ensured to obtain information from all others; 2) decentralized parallel learning \citep{he2018cola, li2017primal, lian2017can, NIPS2018_8274}, where each worker can only gather information from a fraction of its neighbors. Centralized parallel learning requires a supporting network architecture to aggregate all local models or gradients in each iteration. Various implementations are developed for information aggregation in centralized systems. For example, the parameter server \citep{li2014scaling,Abadi:2016:TSL:3026877.3026899}, AllReduce \citep{seide2016cntk,renggli2018sparcml}, adaptive distributed learning, differentially private distributed learning, distributed proximal primal-dual algorithm, non-smooth distributed optimization,  projection-free distributed online learning, and parallel back-propagation for deep learning.

\subsection{Decentralized Parallel Learning}

Unlike centralized learning, decentralized learning does not require obtaining information from all workers in each step. Therefore, the network structure would have fewer restrictions than centralized parallel learning. Due to this high flexibility of the network design, the decentralized parallel learning has been the focus of many recent studies.

In decentralized parallel learning, workers only communicate with their neighbors. There are two main types of decentralized learning algorithms: fixed network topology \citep{he2018cola}, or time-varying \citep{nedic2015distributed,lian2017asynchronous} during training. \citet{yin_d,pmlr-v80-shen18a} shows that the decentralized SGD would converge with a comparable convergence rate to the centralized  algorithm with less needed communication to make large-scale model training feasible. \citet{NIPS2018_8028} provide a systematic analysis of the decentralized learning pipeline.

\subsection{Compressed Communication Distributed Learning}

To further reduce the communication overhead, one promising direction is to compress the variables that are sent between different workers \citep{NIPS2018_7752,NIPS2018_8191}.
Previous works mostly focus on a centralized scenario where the compression is assumed to be unbiased \citep{NIPS2018_7405,pmlr-v80-shen18a,pmlr-v70-zhang17e,NIPS2017_6749,NIPS2018_7519}.
A general theoretical analysis of centralized compressed parallel SGD can be found in \citet{DBLP:conf/nips/AlistarhG0TV17}. Beyond this, some biased compressing methods are also proposed and proven to be quite efficient in reducing the communication cost. One example is the 1-Bit SGD \citep{1-bitexp}, which compresses the entries in gradient vector into $\pm 1$ depends on its sign. The theoretical guarantee of this method is given in \citet{Bernstein:2018aa}.

Recently, another emerging area of compressed distributed learning is decentralized compressed learning. Unlike centralized setting, decentralized learning requires each worker to share the model parameters, instead of the gradients. This differentiating factor could potential invalidate the convergence (see \citet{NIPS2018_7992} for example). Many methods are proposed to solve this problem. One idea is to use the difference of the model updates as the shared information \citep{NIPS2018_7992}. \citet{Koloskova:2019aa} further reduces the communication error by sharing the difference of the model difference, but their result only considers a strongly convex loss function, which is not the case for nowadays deep neural network. \citet{NIPS2018_7992} also propose a extrapolation like strategy to control the compressing level. None of these works employ the error-compensate strategy.

\subsection{Error-Compensated SGD}

Interestingly, an error-compensate strategy for AllReduce 1Bit-SGD is proposed in \citet{1-bitexp}, which can compensate the error in each iteration with quite less accuracy drop than training without the error-compensation. The error-compensation strategy is proved to be able to potentially improve the training efficiency for strongly convex loss function \citep{NIPS2018_7697}.  \citet{pmlr-v80-wu18d} further study an Error-Compensated SGD for quadratic optimization via adding two hyperparameters to compensate the error, but fail to theoretically verify the advantage of using error compensation. Most recently, \citet{tang_ds} give an theoretical analysis showing that in centralized learning, the error-compensation strategy is compatible with an arbitrary compression technique and fundamentally proving why error-compensation admit an improved convergence rate with linear speedup for general non-convex loss function. Even though, their work only study the centralized distributed training. Whether error-compensation strategy can work both theoretically and empirically for decentralized learning remains to be investigated. We will try to answer these questions in this paper.

\section{Algorithm}

\begin{algorithm}[t]\caption{$\OA$}
	\begin{algorithmic}[1]
		\STATE {\bfseries Initialize}: $\x_0$, learning rate $\gamma$, averaging rate $\eta$, initial error $\bdelta = \boldsymbol{0}$, number of total iterations $T$, and weight matrix $W$.

		\FOR {$t=0,\ldots,T-1$}

		\STATE \textbf{(On $i$-th node)}
		\STATE  Randomly sample $\bm{\xi}_t^{(i)}$  and compute local stochastic gradient $\g_t^{(i)} := \nabla F_i(\x_t^{(i)}, \bm{\xi}_t^{(i)})$.
		\STATE Compute the error-compensated variable update $\v_t^{(i)} = \x_t^{(i)} - \gamma \g_t^{(i)} + \bdelta_t^{(i)}$.
		\STATE \emph{Compress} $\v_t^{(i)}$ into $C_\omega[\v_t^{(i)}]$ and update the error $\bdelta_t^{(i)} = \v_t^{(i)} - C_\omega[\v_t^{(i)}]$.

		\STATE \emph{Send} compressed variable $C_\omega[\v_t^{(i)}]$ to the immediate neighbors of $i$-th node.
		\STATE \emph{Receive} $C_\omega[\v_t^{(j)}]$ from all neighbors of $i$-th node where $j \in \mathcal{N}(i)$ (notice that $i\in \mathcal{N}(i)$).

		\STATE Update local model $\x_{t+1}^{(i)} = \x_{t}^{(i)} - \gamma \g_{t}^{(i)} + \eta \sum_{j\in \mathcal{N}(i)} (W_{ij} - I_{ij}) C_\omega[\v_t^{(j)}] $.

		\ENDFOR
		\STATE {\bfseries Output}: $\x$.
	\end{algorithmic}\label{alg:de_ec}
\end{algorithm}

We introduce our proposed $\OA$ algorithm details below.

Consider that there are $n$ workers in the network, where each individual worker (e.g., the $i$-th worker) can only communicate with its  immediate neighbors (denotes worker $i$'s immediate neighbors $\mathcal{N}(i)$). The connection can be represented by a symmetric double stochastic matrix, the weighted matrix $W$. We denote by $W_{ij}$ the weight between the $i$-th and $j$-th worker.

For the $i$-th worker at time $t$, each worker, say worker $i$, need to maintain three main variables: the local model $\x_t^{(i)}$, the compression error $\bdelta_t^{(i)}$, and the error-compensated variable $\v_t^{(t)}$. The updates of those variables follow the following steps:

\begin{enumerate}
\item \textbf{Local Computation:} Update the error-compensated variable by $\v_t^{(i)} = \x_t^{(i)} - \gamma \nabla F_i(\x_t^{(i)}, \bm{\xi}_t^{(i)}) + \bdelta_t^{(i)}$ where $\nabla F_i(\x_t^{(i)}, \bm{\xi}_t^{(i)})$ is the stochastic gradient with randomly sampled $\bm{\xi}_t^{(i)}$, and $\gamma$ is the learning.

\item \textbf{Local Compression:} Compress the error-compensated variable $\v_t^{(i)}$ into $C_\omega[\v_t^{(i)}]$, where $C_\omega[\cdot]$ is the compression operation and update the compression error by $\bdelta_t^{(i)} = \v_t^{(i)} - C_\omega[\v_t^{(i)}]$.

\item \textbf{Global Communication:} Send compressed variable $C_\omega[\v_t^{(i)}]$ to the  neighbors of $i$-th node, namely $\mathcal{N}(i)$.

\item\textbf{Local Update:}  Update local model by $\x_{t+1}^{(i)} = \x_{t}^{(i)} - \gamma \g_{t}^{(i)} + \eta \sum_{j\in \mathcal{N}(i)} (W_{ij} - I_{ij}) C_\omega[\v_t^{(j)}] $, where $I$ denotes the identity matrix, and $\eta$ is the averaging rate.
\end{enumerate}
Finally, the proposed $\OA$ algorithm is summarized in Algorithm \ref{alg:de_ec}. 

\section{Theoretical Analysis}
In this section, we introduce our theoretical analysis of $\OA$. We first prove the updating rule of our algorithm to give an intuitive explanation of how our algorithm works, then we present the final convergence rate.
\subsection{Mathematical formulation}
In order to get a global view of $\OA$, we define
\begin{align*}
&X_t :=  \left[\bm{x}^{(1)}_t, \bm{x}^{(2)}_t, \cdots, \bm{x}_t^{(n)}\right] \in \mathbb{R}^{N\times n},\quad
&G(X_t; \xi_t) :=  \left[\nabla F_1(\bm{x}^{(1)}_t; \bm{\xi}_t^{(1)}), \cdots, \nabla F_n(\bm{x}_t^{(n)}; \bm{\xi}_t^{(n)})\right], \\
&\Delta_t := \left[\bm{\delta}^{(1)}_t,\bm{\delta}^{(2)}_t,\cdots,\bm{\delta}^{(n)}_t\right]\in \mathbb{R}^{N\times n},
&\nabla\overline{ f}(X_t):=  \mathbb{E}_{\xi}G(X_t;\xi_t)\frac{\bm{1}}{n}=\frac{1}{n}\sum_{i=1}^n\nabla f_i\left(\bm{x}_t^{(i)}\right),\\
&\overline{\bm{x}}_t := \frac{1}{n}\sum_{i=1}^n\bm{x}^{(i)}_t.
\end{align*}
Then the updating rule of $\OA$ follows
\begin{align*}
X_{t+1} =& X_t - \gamma G_t  +  \eta C_{\omega}  \left[X_t - \gamma G_t + \Delta_{t-1} \right](W-I)\\
\Delta_{t} = & X_t - \gamma G_t + \Delta_{t-1} - C_{\omega}\left[ X_t - \gamma G_t + \Delta_{t-1} \right],
\end{align*}
which can also be rewritten as
\begin{align*}
X_{t+1} = & \left(X_t - \gamma G_t\right)W_{\eff} + (\Delta_{t-1} - \Delta_t )(W_{\eff}-I),
\end{align*}
where we denote $W_{\eff} = (1-\eta)I + \eta W$.

\subsection{Why $\OA$ is better} We would be able to get a better understanding about why $\OA$ is better than the other compression based decentralized algorithms, by analyzing the closed forms of updating rules. For simplicity, we assume the initial values are $0$. Skipping the detailed derivation, we can obtain the closed forms:
\begin{itemize}
    \item D-PSGD \citep{lian2017can} (without compression):
\[
X_t = - \gamma\sum_{s=0}^{t-1}G_sW_{\eff}^{t-s}.
\]
    \item DCD-PSGD \citep{NIPS2018_7992} (with compression):
    \[
    X_t=  - \gamma\sum_{s=0}^{t-1}G_sW_{\eff}^{t-s} - \eta\sum_{s=0}^{t-2}\Delta_{s} W_{\eff}^{t-s}.
    \]
    \item CHCHO-SGD \citep{Koloskova:2019aa} (with compression):
    \[
    X_t=  - \gamma\sum_{s=0}^{t-1}G_sW_{\eff}^{t-s} - \eta\sum_{s=0}^{t-1}\Delta_{s} (W_{\eff}-I) W_{\eff}^{t-s}.
    \]
        \item Updating rule of $\OA$ (with error-compensated compression):
    \[
    X_t=  - \gamma\sum_{s=0}^{t-1}G_sW_{\eff}^{t-s} -  \eta\Delta_{t-1}(W_{\eff}-I) + \eta\sum_{s=0}^{t-2}\Delta_{s} (W_{\eff}-I)^2 W_{\eff}^{t-s}.
    \]
\end{itemize}
We can notice that if $\Delta_s=0$, that is, there is no compression error, all compression based methods reduce to D-PSGD. Therefore the efficiency of compression based methods lie on the magnitude of $\|\Delta_s\|_F$ (for D-PSGD) or $\|\Delta_s(W_{\text{eff}})-I\|_F$ (for $\OA$ and CHCHO-SGD). 



Using the fact that $W$ is doubly stochastic and symmetric, we have
\begin{align*}
\left\| W_{\eff} - I \right\|_2 =& \left\| \eta W - \eta I \right\|
\leq  \eta(1 - \lambda_n(W)),
\end{align*}
which leads to
\begin{align*}
\left\|\Delta_{s} (W_{\eff}-I)^2 \right\|_F^2\leq & \left\|(W_{\eff} - I)^2 \right\|^2  \left\|\Delta_{s}\right\|^2_F\leq (1 - \lambda_n(W))^4\eta^4\left\|\Delta_{s}\right\|^2_F\quad\text{( $\OA$)} \\
\left\|\Delta_{s} (W_{\eff}-I) \right\|_F^2\leq & \|W_{\eff} - I \|^2\left\|\Delta_{s}\right\|^2_F\leq (1 - \lambda_n(W))^2\eta^2\left\|\Delta_{s}\right\|^2_F\quad\text{(CHOCO-SGD)}\\
\left\|\Delta_{s}  \right\|_F^2\leq & \left\|\Delta_{s}\right\|^2_F.\quad\text{(DCD-PSGD)}
\end{align*}
which means our algorithm could be able to significantly reducing the influence of the history compressing error by controlling $\eta$. Actually, it will be clear soon (from our analysis), $\eta$ has to be small enough if we compress the information very aggressive. 

Before we introducing the final convergence rate of $\OA$, we first introduce some assumptions that is used for theoretical analysis.
\begin{assumption}
\label{ass:global}
Throughout this paper, we make the following commonly used assumptions:
\begin{enumerate}
  \item \textbf{Lipschitzian gradient:} All function $f_i(\cdot)$'s are with $L$-Lipschitzian gradients.
  \item \textbf{Symmetric double stochastic matrix:} The  weighted matrix $W$ is a real double stochastic matrix that satisfies $W=W^{\top}$ and $W\bm{1}=W$.
  \item \textbf{Spectral gap:} Given the symmetric doubly stochastic matrix $W$, we assume that $\lambda_2(W)\leq 1$.
  \item \textbf{Start from 0:} We assume $X_0 = 0$. This assumption simplifies the proof w.l.o.g.
  \item \textbf{Bounded variance:} Assume the variance of stochastic gradient to be bounded 
{\begin{align*}
    \mathbb{E}_{\xi\sim \mathcal{D}_i} \left\| \nabla F_i (\bm{x}; \xi) - \nabla f_i (\bm{x})\right\|^2 \leqslant &  \sigma^2, \quad \text{Inner variance}
    \\
     {1\over n}\sum_{i=1}^n\left\| \nabla f_i (\bm{x})-\nabla f (\bm{x})\right\|^2 \leqslant &  \zeta^2, \quad \forall i, \forall \bm{x}, \quad \text{Outer variance}
\end{align*}}
\item  \textbf{Bounded Signal-to-noise factor:} The magnitude of the compression error is assumed to be bounded by the original vector's magnitude:
\begin{align*}
\mathbb E\left\|C_{\omega}[\bm{x}] - \bm{x} \right\|^2 \leq \alpha^2\|\bm{x}\|^2,\quad\alpha\in[0,1),\forall \bm{x}.
\end{align*}
 \end{enumerate}
\end{assumption}
Notice that these assumptions are quite standard for analyzing decentralized algorithms in previous works \citet{lian2017can}. 
Unlike \citet{Koloskova:2019aa}, we do not require the gradient to be bounded ($\|\nabla f(\bm{x})\|\leq G^2  $), which makes our result more applicable to a general case. It is worth pointing out that many of the previous compressing strategies satisfy the bounded signal-to-noise factor assumption, such as GSpar, random sparsification \citep{NIPS2018_7405}, top-$k$ sparsification and random quantization \citep{Koloskova:2019aa}.

Now we are ready to show the main theorem of $\OA$.

\begin{theorem}\label{theo}
For $\OA$, if the averaging rate $\eta$ and learning rate $\gamma$ satisfies
\begin{align*}
&\eta \leq  \min\left\{\frac{1}{2},\frac{\alpha^{-\frac{2}{3}} - 1}{4}\right\}\\
&1 - 3C_2L^2\gamma^2 \geq  0\\
&\gamma \leq \frac{1}{L},
\end{align*}
then we have
\begin{align*}
&\left(\frac{\gamma}{2} -3C_3\gamma^3 \right)\sum_{t=0}^T \mathbb E \left\|\nabla f(\overline{\bm{x}}_t)\right\|^2\\
\leq & \mathbb E f\left(\overline{\bm{x}}_{0} \right) - \mathbb E f\left(\overline{\bm{x}}_{T} \right)  + \left(\frac{L\gamma^2}{2n} + C_3\gamma^3\right)\sigma^2T + 3C_3\gamma^3\zeta^2,
\end{align*}
where
\begin{align*}
C_0 := & \eta( 1- \lambda_n(W))\\
C_1 := & \frac{\alpha^2}{\left(1 - \alpha^2\left(1 + C_0\right)^2(1 + 2C_0) \right)C_0^2} \\
 C_2 := & \frac{3}{\eta^2(1-\lambda_2(W))^2} + 6C_1\\
 C_3 := &\frac{C_2L^2}{2 -6C_2L^2\gamma^2}.
\end{align*}
\end{theorem}

To make the result more clear, we choose the learning rate $\gamma$ appropriately in the following:
\begin{corollary}\label{coro}
According to Theorem~\ref{theo}, choosing $\gamma = \frac{1}{3L\sqrt{C_2} + \sigma\sqrt{\frac{T}{n}} + \zeta^{\frac{2}{3}}T^{\frac{1}{3}}} $ and $\eta = \min\left\{\frac{1}{2},\frac{\alpha^{-\frac{2}{3}} - 1}{4}\right\} $, we have the following convergence rate for $\OA$
\begin{align*}
\frac{1}{T}\sum_{t=0}^T E \left\|\nabla f(\overline{\bm{x}}_t)\right\|^2 \lesssim & \left(\frac{1}{\sqrt{nT}} + \frac{C_2}{T}\right)\sigma +
\frac{C_2\zeta^{\frac{2}{3}}}{T^{\frac{2}{3}}} + \frac{1 }{T},
\end{align*}
where we have
\begin{align*}
C_2 \lesssim \frac{1}{(1 - \lambda_2(W))^2}\left ( 1+  \frac{\alpha^2}{\left( \alpha^{-\frac{2}{3}} - 1 \right)} \right).
\end{align*}
where $L$ is treated to be a constant.
\end{corollary}
This result suggests that:

\textbf{Linear speedup} The asymptotical convergence rate of $\OA$ is {$O \left(1/\sqrt{nT}\right)$}, which is exactly the same with Centralized Parallel SGD.

\textbf{Consistence with D-PSGD} Setting $\alpha = 0$, our $\OA$ reduces to the standard decentralized SGD (D-PSGD), our result indicates that the convergence rate admits $O \left(\frac{\sigma}{\sqrt{nT}} + \frac{\zeta^{\frac{2}{3}}}{T^{\frac{2}{3}}}\right)$, which is slightly better than the previous one $O \left(\frac{\sigma}{\sqrt{nT}} + \frac{n^{\frac{2}{3}}\zeta^{\frac{2}{3}}}{T^{\frac{2}{3}}}\right)$ in \citet{lian2017can}.

\textbf{Superiority over DCD-PSD and ECD-PSGD}
By using the error-compensate strategy, we prove that the the dependence to the compression ratio $\alpha$  is $O\left( \frac{\sigma\alpha^2}{ T} + \frac{\zeta^{\frac{2}{3}}\alpha^2}{ T^{\frac{2}{3}}} \right) $, and is robust to any compression operator with $\alpha\in[0,1)$. Whereas the previous work \citep{NIPS2018_7992} only ensure a dependence $O\left( \frac{\sigma\alpha^2}{ \sqrt{T}} + \frac{\zeta^{\frac{2}{3}}\alpha^2}{ T^{\frac{2}{3}}} \right) $  and there is a restriction that $\alpha\leq (1 - \lambda_2(W))^2$.

\textbf{Balance between Communication and Computation:} Our result indicates that when $\alpha \to 1$, we need to set $\eta$ to be small enough to ensure the convergence of our algorithm. However, too small $\eta$ leads to a slower convergence rate, which means more iterations are needed for training. Our result could be a guidance for balancing the communication cost and computation cost for decentralized learning under different situations.
\section{Experiments}

In this section, we further demonstrate the superiority of the $\OA$ algorithm through 
an empirical study. Under aggressive compression ratios, we show the epoch-wise 
convergence results of $\OA$ and compare them with centralized
SGD (AllReduce), decentralized SGD without compression (D-PSGD \citep{lian2017can}),
and other existing decentralized compression algorithms (ECP-PSGD, DCP-PSGD
\citep{NIPS2018_7992}, and Choco-SGD \citep{Koloskova:2019aa}).

\subsection{Experimental Setup}

\paragraph{Dataset} We benchmark the algorithms with a standard image
classification task: CIFAR10 using ResNet-20. This dateset has a training
set of 50,000 images and a test set of 10,000 images, where each image is given
one of the 10 labels.

\paragraph{Communication} The communication is implemented based on NVIDIA
NCCL for AllReduce and gloo for all other algorithms. The workers are connected in a ring network so that each worker has two neighbors.

\paragraph{Compression} We use bit compression for all the compression
algorithms, where every element of the tensor sent out or received by a worker
is compressed to 4 bits or 2 bits. A number representing the Euclidean norm of
the original tensor is also sent together, so that when a worker receives the
compressed tensor, it can scale it to get a tensor with the same Euclidean norm
as the uncompressed tensor.

\paragraph{Hardware} The experiments are done on 8 2080Ti GPUs, where each GPU
is treated as a single worker in the algorithm.

\subsection{Convergence Efficiency}

We compare the algorithms with 4bit and 2bit compression, and the convergence
results are shown in Figure \ref{fig:4bit}. We use a batch size of 128 and tune
the learning rate in a $\{1, 0.5, 0.1, 0.01\}$ set for each algorithm. The
learning rate is decreased by 5x every 60 epochs. 

Note that, compared with uncompressed methods, 4bit compression saves 7/8 of the communication cost 
while 2bit compression reduces the cost by 93.75\%. The results
show that under 4bit compression, $\OA$ converges similar to uncompressed
algorithms. $\OA$ is slightly better
than Choco-SGD \citep{Koloskova:2019aa}, and converges much faster than DCP-PSGD and ECD-PSGD.
Note that DCP-PSGD does not converge under 4bit compression (as stated in the
original paper \citep{NIPS2018_7992}), and ECD-PSGD does not converge under 2bit compression.
Under 2bit compression, $\OA$ converges slower than uncompressed algorithms, but is
still faster than any other compression algorithms. The reason is that the less
bits we have, the more we gain from the error compensation in $\OA$, and
consequently $\OA$ gets a faster convergence.

The experimental results further confirm that $\OA$ achieves outstanding communication cost reduction (up to 93.75\%) without much sacrifice on convergence.

\begin{figure}
  \centering
  \includegraphics[width=0.49\textwidth]{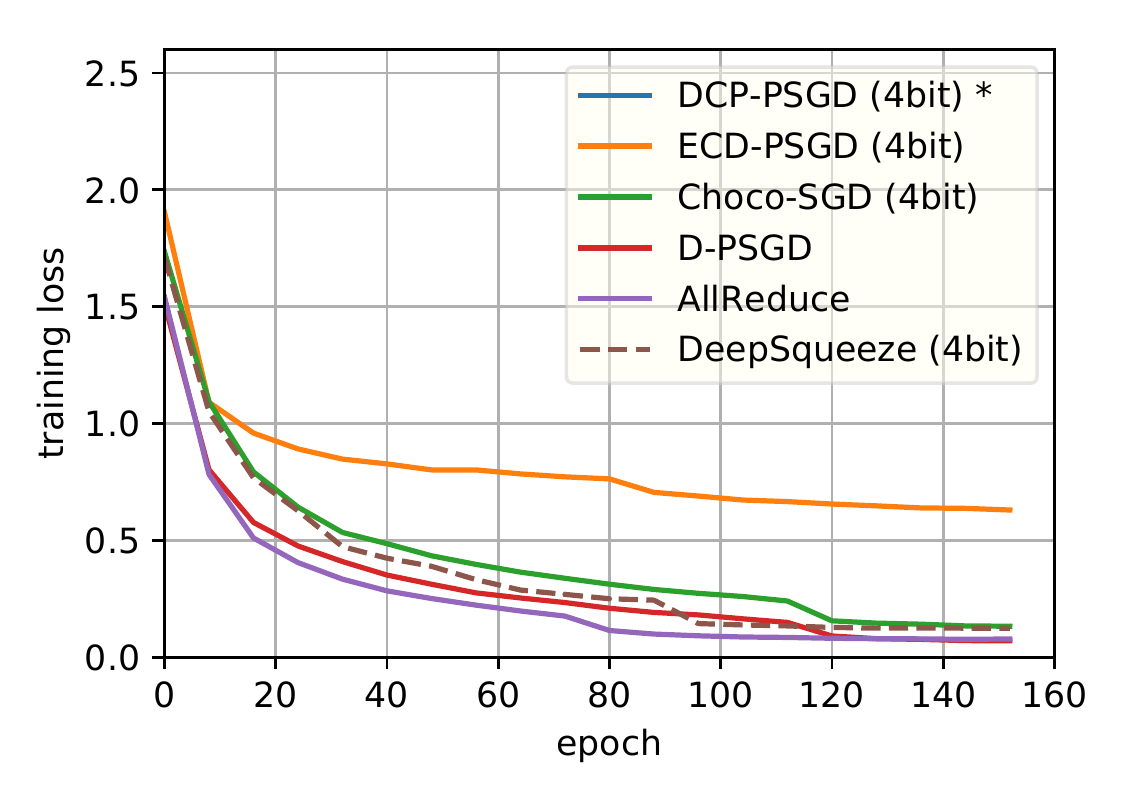}
  \includegraphics[width=0.49\textwidth]{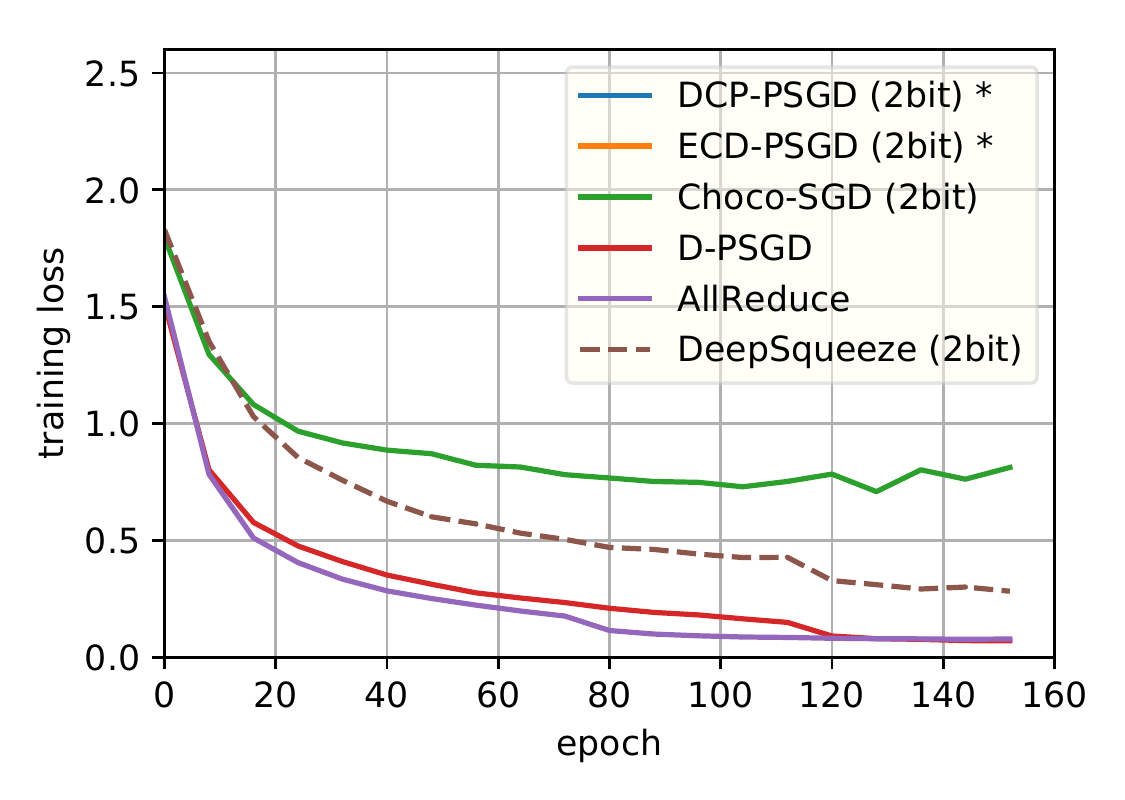}
  \caption{Epoch-wise convergence for the algorithms with 8 workers connected on a ring network. 4bit compression is on the left and 2bit compression is on the right. Note that DCP-PSGD (marked
    with *) does not
    converge below 4bit compression as stated in the original paper \citep{NIPS2018_7992},
    and ECP-PSGD does not converge under 2bit compression, so they are not shown in the figure.\label{fig:4bit}}
\end{figure}

\section{Conclusion}
This paper applies the error-compensated compression strategy, which receives big successes in centralized training, to the decentralized training. Through elaborated algorithm, the proposed $\OA$ algorithms is proven to admit higher tolerance on the compression ratio than existing compression based decentralized training. Both theoretical analysis and empirical study validate our algorithm and theorem.


\newpage

\bibliographystyle{abbrvnat}
\bibliography{reference}

\newpage

\newpage

\begin{center}
    \LARGE{ Supplementary Material}
\end{center}

We provide the proofs to our theorems in the supplemental material.

\section*{Preliminaries} Below we would use some basic properties of matrix and vector.
\begin{itemize}
\item For any $A\in \mathbb R^{d\times n} $ and $P \in \mathbb R^{n\times n} $, if $PP^{\top}  = I$, then we have
\begin{align*}
\|AP\|_F^2 = \|AP^{\top}\|_F^2 = \|A\|_F^2.
\end{align*}
\item For any two vectors $\bm{a} $ and $\bm{b}$, we have
\begin{align*}
\|\bm{a} + \bm{b} \|^2 \leq (1 + \alpha)\|\bm{a}\|^2 + \left( 1 + \frac{1}{\alpha} \right)\|\bm{b}\|^2, \quad\forall \alpha > 0.
\end{align*}
\item For any two matrices $A $ and $B$, we have
\begin{align*}
\|A + B \|^2_F \leq (1 + \alpha)\|A\|^2_F + \left( 1 + \frac{1}{\alpha} \right)\|B\|^2_F, \quad\forall \alpha > 0.
\end{align*}
\end{itemize}

Before presenting our proof, we first rewrite our updating rule  as follows
\begin{align*}
\Delta_t =& (X_t - \gamma G_t) + \Delta_{t-1} - C_{\omega}\left[ (X_t - \gamma G_t) + \Delta_{t-1} \right]\\
X_{t+1} =& (X_t - \gamma G_t) W_{\eff}+ \eta (\Delta_{t-1} - \Delta_t) (W_{\eff}-I) \numberthis\label{general:updating_rule}
\end{align*}
where 
\begin{align*}
W_{\eff} := (1-\eta)I + \eta W.
\end{align*}

\section*{Proof of Lemma~\ref{lemma:key_theo_lemma}} \label{sec:proof_key_lemma}

Lemma \ref{lemma:key_theo_lemma} is the key lemma for proving our Theorem \ref{theo}. 

\begin{lemma}\label{lemma:key_theo_lemma}
For  algorithm that admits the updating rule \eqref{general:updating_rule}, if 
\begin{align*}
W_{\eff} \succeq & 0,\\
(2 - \lambda_n )^2(3 - 2\lambda_n) \leq & \frac{1}{\alpha^2},\\
 1 -3C_5L^2\gamma^2 \geq & 0,\\
 L \leq & \frac{1}{L}
\end{align*}
then we have
\begin{align*}
&\left(\frac{\gamma}{2} -\frac{3C_5L^2 \gamma^3}{2 -6C_5L^2\gamma^2} \right)\sum_{t=0}^T E \left\|\nabla f(\overline{\bm{x}}_t)\right\|^2\\
\leq & \mathbb E f\left(\overline{\bm{x}}_{0} \right) - \mathbb E f\left(\bm{x}^* \right)  + \left(\frac{L\gamma^2}{2} + \frac{C_5L^2\gamma^3}{2 -6C_5L^2\gamma^2)}\right)\sigma^2T + \frac{3C_5L^2\gamma^3\zeta^2T}{2 -6C_5L^2\gamma^2}.
\end{align*}
where 
\begin{align*}
C_4 := & \frac{\alpha^2}{(1 - \alpha^2(2-\lambda_n)^2(3 -2\lambda_n) )(1-\lambda_n)^2} \\
 C_5 := & \frac{3}{(1-\lambda_2)^2} + 6C_4,
\end{align*}
and $\lambda_n = \lambda_n(W_{\eff})$ and $\lambda_2 = \lambda_2(W_{\eff})$ for short.
\end{lemma}

We outline our proof of Lemma \ref{lemma:key_theo_lemma} as follows.

The most challenging part of a decentralized algorithm, unlike the centralized algorithm, is that we need to ensure the local model on each node to converge to the average value $\overline{\bm{x}}_t$. This is because
\begin{align*}
&\mathbb E f\left(\overline{\bm{x}}_{t+1} \right) - \mathbb E f\left(\overline{\bm{x}}_{t} \right)\\
\leq & -\frac{\gamma}{2}\mathbb E \left\|\nabla \of(X_t)\right\|^2 -\left(\frac{\gamma}{2} -\frac{L\gamma^2}{2}\right) \mathbb E \left\|\nabla \of(X_t)\right\|^2 + \frac{\gamma L^2}{2n}\mathbb E \left\|X_t(I - A_n) \right\|^2_F+ \frac{L\gamma^2\sigma^2}{2n}.\numberthis\label{lemma:why_bound_diff}
\end{align*}

So  we first prove that
\begin{align*}
\sum_{t=0}^{T-1}\sum_{i=1}^n \left\|\overline{\bm{x}}_t - \bm{x}_t^{(i)} \right\|^2 = &  \sum_{t=0}^{T-1} \mathbb E\left\|X_t(I - A_n)\right\|_F^2\\
\leq & \frac{3\gamma^2}{(1-\lambda_2)^2}\sum_{t=0}^{T-2}\left\|G_t\right\|^2_F + 6\sum_{t=0}^{T-1} \left\|\Delta_{t}(W_{\eff}-I) \right\|_F^2,
\end{align*}
and the compressing error can be upper bounded by
\begin{align*}
\sum_{t=0}^{T-1}\mathbb E\left\|\Delta_{t}(W_{\eff}-I)\right\|^2 \leq C_4\gamma^2\sum_{t=0}^{T-1} \mathbb E\left\|G_t \right\|^2_F.
\end{align*}
With these two important observations, we finally prove that
\begin{align*}
\sum_{t=0}^{T-1} \mathbb E\left\|X_t(I - A_n)\right\|_F^2 \leq \frac{C_5\gamma^2n\sigma^2 T}{1 -3C_5L^2\gamma^2} + \frac{3nC_5\gamma^2\zeta^2 T}{1 -3C_5L^2\gamma^2} + \sum_{t=0}^{T-1}\frac{3nC_5 \gamma^2\mathbb{E}\left\|\nabla f(\overline{\bm{x}}_t)\right\|^2}{1 -3C_5L^2\gamma^2}.
\end{align*}
So taking the equation above into \eqref{lemma:why_bound_diff} we could directly prove Theorem~\ref{lemma:key_theo_lemma}.

Now we introduce the detailed proof for Lemma \ref{lemma:key_theo_lemma}.

\paragraph{Proof of Lemma~\ref{lemma:key_theo_lemma}}
\begin{proof}
The updating rule of $\OA$ can be written as
\begin{align*}
\Delta_t =&  -\gamma G_t + X_t  + \Delta_{t - 1} - C_\omega\left[  -\gamma G_t + X_t  + \Delta_{t - 1} \right],\\
X_{t+1} =& (X_t - \gamma G_t) W_{\eff} + (\Delta_{t-1} - \Delta_t) (W_{\eff}-I),
\end{align*}

The updating rule of $\overline{\bm{x}}_t$ can be conducted using the equation above
\begin{align*}
\overline{\bm{x}}_{t+1} =& \overline{\bm{x}}_t - \gamma\bG_t.
\end{align*}
So we have
\begin{align*}
&\mathbb E f\left(\overline{\bm{x}}_{t+1} \right) - \mathbb E f\left(\overline{\bm{x}}_{t} \right)\\
\leq & - \gamma\mathbb E \left\langle \bG_t,\nabla f(\overline{\bm{x}}_t) \right\rangle + \frac{L\gamma^2}{2}\mathbb E \left\|\bG_t\right\|^2\\
= & - \gamma\mathbb E \left\langle \nabla \of(X_t),\nabla f(\overline{\bm{x}}_t) \right\rangle + \frac{L\gamma^2}{2}\mathbb E \left\|\bG_t - \nabla\of(X_t)\right\|^2 + \frac{L\gamma^2}{2}\mathbb E \left\|\nabla\of(X_t)\right\|^2\\
= & - \gamma\mathbb E \left\langle\nabla  \of(X_t),\nabla f(\overline{\bm{x}}_t) \right\rangle + \frac{L\gamma^2}{2}\mathbb E \left\|\nabla\of(X_t)\right\|^2 + \frac{L\gamma^2\sigma^2}{2n}.
\end{align*}
Using
\begin{align*}
\left\langle\nabla  \of(X_t),\nabla f(\overline{\bm{x}}_t) \right\rangle = \frac{1}{2}\left( \left\|\nabla \of(X_t)\right\|^2 + \left\|\nabla f(\overline{\bm{x}}_t)\right\|^2 - \left\|\nabla \of(X_t) -\nabla f(\overline{\bm{x}}_t)\right\|^2\right),
\end{align*}
then the equation above becomes
\begin{align*}
&\mathbb E f\left(\overline{\bm{x}}_{t+1} \right) - \mathbb E f\left(\overline{\bm{x}}_{t} \right)\\
\leq &  - \frac{\gamma}{2}\left( \mathbb E \left\|\nabla \of(X_t)\right\|^2 + \mathbb E \left\|\nabla f(\overline{\bm{x}}_t)\right\|^2 - \mathbb E \left\|\nabla \of(X_t) -\nabla f(\overline{\bm{x}}_t)\right\|^2\right) \\
&+ \frac{L\gamma^2}{2}\mathbb E \left\|\nabla\of(X_t)\right\|^2 + \frac{L\gamma^2\sigma^2}{2n}\\
= & -\frac{\gamma}{2}\mathbb E \left\|\nabla f(\overline{\bm{x}}_t)\right\|^2 -\left(\frac{\gamma}{2} -\frac{L\gamma^2}{2}\right) \mathbb E \left\|\nabla \of(X_t)\right\|^2 + \frac{\gamma}{2}\mathbb E \left\|\nabla \of(X_t) -\nabla f(\overline{\bm{x}}_t)\right\|^2+ \frac{L\gamma^2\sigma^2}{2n}.\numberthis\label{supp:theo_eq1}
\end{align*}
The difference between $\nabla \of(X_t) $ and $\nabla f(\overline{\bm{x}}_t)$ can be upper bounded by
\begin{align*}
\mathbb E\left\|\nabla \of(X_t) -\nabla f(\overline{\bm{x}}_t)\right\|^2 = & \mathbb E\left\|\frac{1}{n}\sum_{i=1}^n \nabla f_i(\bm{x}_t^{(i)}) - \nabla f_i(\overline{\bm{x}}_t) \right\|^2\\
\leq &\frac{1}{n^2}\mathbb E\left\|\sum_{i=1}^n \nabla f_i(\bm{x}_t^{(i)}) - \nabla f_i(\overline{\bm{x}}_t) \right\|^2\\
\leq & \frac{1}{n}\sum_{i=1}^n\mathbb E\left\| \nabla f_i(\bm{x}_t^{(i)}) - \nabla f_i(\overline{\bm{x}}_t) \right\|^2\\
\leq & \frac{L^2}{n}\sum_{i=1}^n\mathbb E\left\| \bm{x}_t^{(i)} - \overline{\bm{x}}_t \right\|^2\\
= & \frac{L^2}{n}\mathbb E\left\|X_t(I - A_n) \right\|^2_F.
\end{align*}
So \eqref{supp:theo_eq1} becomes
\begin{align*}
&\mathbb E f\left(\overline{\bm{x}}_{t+1} \right) - \mathbb E f\left(\overline{\bm{x}}_{t} \right)\\
\leq & -\frac{\gamma}{2}\mathbb E \left\|\nabla f(\overline{\bm{x}}_t)\right\|^2 -\left(\frac{\gamma}{2} -\frac{L\gamma^2}{2}\right) \mathbb E \left\|\nabla \of(X_t)\right\|^2 + \frac{\gamma L^2}{2n}\mathbb E \left\|X_t(I - A_n) \right\|^2_F+ \frac{L\gamma^2\sigma^2}{2n}.
\end{align*}
Continuing using Lemma~\ref{lemma:bound_diff_X} to give upper bound for $\mathbb E \left\|X_t(I - A_n) \right\|^2_F$, we have
\begin{align*}
&\mathbb E f\left(\overline{\bm{x}}_{t+1} \right) - \mathbb E f\left(\overline{\bm{x}}_{t} \right)\\
\leq & -\frac{\gamma}{2}\mathbb E \left\|\nabla \of(X_t)\right\|^2 -\left(\frac{\gamma}{2} -\frac{L\gamma^2}{2}\right) \mathbb E \left\|\nabla \of(X_t)\right\|^2 + \frac{\gamma L^2}{2n}\mathbb E \left\|X_t(I - A_n) \right\|^2_F+ \frac{L\gamma^2\sigma^2}{2n}\\
\leq & -\frac{\gamma}{2}\mathbb E \left\|\nabla f(\overline{\bm{x}}_t)\right\|^2 -\left(\frac{\gamma}{2} -\frac{L\gamma^2}{2}\right) \mathbb E \left\|\nabla \of(X_t)\right\|^2 + \frac{L\gamma^2\sigma^2}{2n}\\
& + \frac{C_5L^2\gamma^3\sigma^2}{2 -6C_5L^2\gamma^2)} + \frac{3C_5L^2\gamma^3\zeta^2}{2 -6C_5L^2\gamma^2} + \frac{3C_5L^2 \gamma^3\mathbb{E}\left\|\nabla f(\overline{\bm{x}}_t)\right\|^2}{2 -6C_5L^2\gamma^2},
\end{align*}
which can be rewritten as
\begin{align*}
&\left(\frac{\gamma}{2} -\frac{3C_5L^2 \gamma^3}{2 -6C_5L^2\gamma^2} \right) E \left\|\nabla f(\overline{\bm{x}}_t)\right\|^2\\ \leq & \mathbb E f\left(\overline{\bm{x}}_{t} \right) - \mathbb E f\left(\overline{\bm{x}}_{t+1} \right) - \left(\frac{\gamma}{2} -\frac{L\gamma^2}{2}   \right)\mathbb{E}\left\|\nabla \of(X_t)\right\|^2 + \left(\frac{L\gamma^2}{2n} + \frac{C_5L^2\gamma^3}{2 -6C_5L^2\gamma^2)}\right)\sigma^2\\
& + \frac{3C_5L^2\gamma^3\zeta^2}{2 -6C_5L^2\gamma^2}.
\end{align*}
Summing up the equation above from $t=0$ to $t=T-1$ we get
\begin{align*}
&\left(\frac{\gamma}{2} -\frac{3C_5L^2 \gamma^3}{2 -6C_5L^2\gamma^2} \right)\sum_{t=0}^T E \left\|\nabla f(\overline{\bm{x}}_t)\right\|^2\\
\leq & \mathbb E f\left(\overline{\bm{x}}_{0} \right) - \mathbb E f\left(\overline{\bm{x}}_{T} \right)  + \left(\frac{L\gamma^2}{2n} + \frac{C_5L^2\gamma^3}{2 -6C_5L^2\gamma^2}\right)\sigma^2T+ \frac{3C_5L^2\gamma^3\zeta^2T}{2 -6C_5L^2\gamma^2}.\quad\text{($\gamma\leq \frac{1}{L}$)}
\end{align*}

This concludes the proof.
\end{proof}

Below are some critical lemmas for the proof of Lemma \ref{lemma:key_theo_lemma}. 

\begin{lemma}\label{lemma:seq}
Given two non-negative sequences $\{a_t\}_{t=1}^{\infty}$ and $\{b_t\}_{t=1}^{\infty}$ that satisfying
\begin{equation}
a_t =  \sum_{s=1}^t\rho^{t-s}b_{s}, \numberthis \label{eqn1}
\end{equation}
with $\rho\in[0,1)$, we have
\begin{align*}
D_k:=\sum_{t=1}^{k}a_t^2 \leq & 
\frac{1}{(1-\rho)^2} \sum_{s=1}^kb_s^2.
\end{align*}
\end{lemma}

\begin{proof}
From the definition, we have
\begin{align*}
S_k= & \sum_{t=1}^{k}\sum_{s=1}^t\rho^{t-s}b_{s}
=  \sum_{s=1}^{k}\sum_{t=s}^k\rho^{t-s}b_{s}
=  \sum_{s=1}^{k}\sum_{t=0}^{k-s}\rho^{t}b_{s}
\leq  \sum_{s=1}^{k}{b_{s}\over 1-\rho}, \numberthis \label{eqn3}\\
D_k=  & \sum_{t=1}^{k}\sum_{s=1}^t\rho^{t-s}b_{s}\sum_{r=1}^t\rho^{t-r}b_{r}\\
= & \sum_{t=1}^{k}\sum_{s=1}^t\sum_{r=1}^t\rho^{2t-s-r}b_{s}b_{r} \\
\leq &  \sum_{t=1}^{k}\sum_{s=1}^t\sum_{r=1}^t\rho^{2t-s-r}{b_{s}^2+b_{r}^2\over2}\\
= & \sum_{t=1}^{k}\sum_{s=1}^t\sum_{r=1}^t\rho^{2t-s-r}b_{s}^2 \\
\leq  & {1\over 1-\rho}\sum_{t=1}^{k}\sum_{s=1}^t\rho^{t-s}b_{s}^2\\
\leq & {1\over (1-\rho)^2}\sum_{s=1}^{k}b_{s}^2. \quad \text{(due to \eqref{eqn3})}
\end{align*}
\end{proof}

\begin{lemma}\label{lemma:bound_pow2}
For any matrix sequence $\{M_t\}$ and positive integer constant $m\in\{1,2,\cdots\}$, we have
\begin{align*}
\sum_{t=0}^{T-1}\left\| \sum_{s=0}^t M_s(W_{\eff} - I)^mW_{\eff}^{t-s} \right\|^2_F \leq  (1- \lambda_n)^{2m-2}\sum_{t=0}^{T-1}\left\|  M_t \right\|^2_F.
\end{align*}
\end{lemma}

\begin{proof}
Since $W_{\eff}$ is symmetric, so we decompose it as $W_{\eff} = P\Lambda P^{\top} $.  Denote $B_s:= M_s P^{\top}$ and $\bm{b}_s^{(i)}$ be the $i$th column of $B_s$, we have
\begin{align*}
\left\| \sum_{s=0}^t M_s(W_{\eff} - I)^mW_{\eff}^{t-s} \right\|^2_F = &
\left\| \sum_{s=0}^t M_sP^{\top}(\Lambda - I)^m\Lambda^{t-s}P \right\|^2_F\\
= & \left\| \sum_{s=0}^t M_sP^{\top}(\Lambda - I)^m\Lambda^{t-s} \right\|^2_F\\
= & \left\| \sum_{s=0}^t B_s(\Lambda - I)^m\Lambda^{t-s} \right\|^2_F\\
= & \sum_{i=1}^n \left\| \sum_{s=0}^t \bm{b}_s^{(i)}(\lambda_i - 1)^m\lambda^{t-s}_i \right\|^2\\
= & \sum_{i=1}^n (\lambda_i - 1)^{2m} \left\| \sum_{s=0}^t \bm{b}_s^{(i)}\lambda^{t-s}_i \right\|^2.
\end{align*}
So we have
\begin{align*}
&\sum_{t=0}^{T-1}\left\| \sum_{s=0}^t M_s(W_{\eff} - I)^mW_{\eff}^{t-s} \right\|^2_F\\
= & \sum_{t=0}^{T-1}\sum_{i=1}^n (\lambda_i - 1)^{2m} \left\| \sum_{s=0}^t \bm{b}_s^{(i)}\lambda^{t-s}_i \right\|^2\\
= & \sum_{i=1}^n (\lambda_i - 1)^{2m} \sum_{t=0}^{T-1}\left\| \sum_{s=0}^t \bm{b}_s^{(i)}\lambda^{t-s}_i \right\|^2\\
\leq & \sum_{i=1}^n (\lambda_i - 1)^{2m} \left( \frac{1}{\left(1-\lambda_i \right)^2} \sum_{t=0}^{T-1}\left\|  \bm{b}_t^{(i)} \right\|^2 \right)\quad\text{(using Lemma~\ref{lemma:seq})} \\
= & \sum_{i=1}^n \frac{(\lambda_i - 1)^{2m}}{\left(1-\lambda_i \right)^2} \sum_{t=0}^{T-1}\left\|  \bm{b}_t^{(i)} \right\|^2 \\
\leq & (1- \lambda_n)^{2m-2}\sum_{t=0}^{T-1}\sum_{i=1}^n\left\|  \bm{b}_t^{(i)} \right\|^2\\
= & (1- \lambda_n)^{2m-2}\sum_{t=0}^{T-1}\left\|  M_t \right\|^2_F.
\end{align*}

\end{proof}

\begin{lemma}\label{lemma:bound_compress_error}
For AlgXX, we have the following upper bound for $\mathbb E\|\Delta_t (W_{\eff}-I)\|^2_F$:
\begin{align*}
\sum_{t=0}^{T-1}\mathbb E\left\|\Delta_{t}(W_{\eff}-I)\right\|^2 \leq C_4\gamma^2\sum_{t=0}^{T-1} \mathbb E\left\|G_t \right\|^2_F,
\end{align*}
when
\begin{align*}
(2 - \lambda_n )^2(3 - 2\lambda_n) \leq \frac{1}{\alpha^2}.
\end{align*}

\end{lemma}

\begin{proof}
From Assumption XXX, for $\Delta_t$, we have
\begin{align*}
&\mathbb E\|\Delta_t (W_{\eff}-I)\|^2_F\\
 \leq &\mathbb E \alpha^2 \left\| \left (X_t - \gamma G_t + \Delta_{t-1} \right)(W_{\eff}-I) \right\|^2_F\\
= & \alpha^2\mathbb E\left\|- \gamma \sum_{s=0}^{t}G_sW_{\eff}^{t-s}(W_{\eff}-I) + \sum_{s=0}^{t-2} \Delta_{s} (W_{\eff}-I)^3W_{\eff}^{t-2-s} - \Delta_{t-1}(W_{\eff}-I)^2 + \Delta_{t-1}(W_{\eff}-I) \right\|^2_F\\
\leq & \alpha^2\mathbb E\left(1 + \frac{1}{\beta_1} \right) \left\|- \gamma \sum_{s=0}^{t}G_sW_{\eff}^{t-s}(W_{\eff}-I) + \sum_{s=0}^{t-2} \Delta_{s} (W_{\eff}-I)^3W_{\eff}^{t-2-s}  \right\|^2_F\\
& + \alpha^2(1 + \beta_1) \mathbb E\left\|\Delta_{t-1}(2I - W_{\eff})(W_{\eff}-I)\right\|^2_F\\
\leq & \alpha^2 \left(1 + \frac{1}{\beta_1} \right)(1 + \beta_2 )\mathbb E\left\| \sum_{s=0}^{t-2} \Delta_{s} (W_{\eff}-I)^3W_{\eff}^{t-2-s}  \right\|^2_F\\
& + \alpha^2\left(1 + \frac{1}{\beta_1} \right)\left(1 + \frac{1}{\beta_2} \right)\mathbb E\left\|- \gamma \sum_{s=0}^{t}G_sW_{\eff}^{t-s}(W_{\eff}-I)   \right\|^2_F\\
& + \alpha^2(1 + \beta_1)( 2 - \lambda_n )^2\mathbb E\left\|\Delta_{t-1}(W_{\eff}-I)\right\|^2.\numberthis\label{supp:key_eq1}
\end{align*}
Defining $\Delta_{-1} = \Delta_{-2} = 0$, and
bound $\left\| \sum_{s=0}^{t-2} \Delta_{s} (W_{\eff}-I)^3W_{\eff}^{t-2-s}  \right\|^2_F$ in the equation above using Lemma~\ref{lemma:bound_pow2}, which leads to
\begin{align*}
\sum_{t=0}^{T-1}\left\| \sum_{s=-2}^{t-2} \Delta_{s} (W_{\eff}-I)^3W_{\eff}^{t-2-s}  \right\|^2_F = & \sum_{t=-2}^{T-3}\left\| \sum_{s=-2}^{t} \Delta_{s} (W_{\eff}-I)^3W_{\eff}^{t-s}  \right\|^2_F \\
= & \sum_{t=0}^{T-3}\left\| \sum_{s=0}^{t} \Delta_{s} (W_{\eff}-I)^3W_{\eff}^{t-s}  \right\|^2_F\quad\text{($\Delta_{-1} = \Delta_{-2} = 0$)}\\
\leq & (1-\lambda_n )^2\sum_{t=0}^{T-3}\left\|  \Delta_{t} (W_{\eff}-I)  \right\|^2_F. \numberthis\label{supp:key_eq2}
\end{align*}
Setting $ \beta_1 = \beta_2 =  1-\lambda_n $, summing up both sides of equation \eqref{supp:key_eq1} from $t=0$ to $t=T-1$ we get
\begin{align*}
& \sum_{t=0}^{T-1} \mathbb E\|\Delta_t (W_{\eff}-I)\|^2_F\\
\leq & \alpha^2\left(  1 + \frac{1}{1-\lambda_n} \right )(2 - \lambda_n)\sum_{t=0}^{T-1}\mathbb E\left\| \sum_{s=0}^{t-2} \Delta_{t} (W_{\eff}-I)^3W_{\eff}^{t-2-s}  \right\|^2_F \\
&+ \alpha^2(2 - \lambda_n)^3\sum_{t=0}^{T-1}\mathbb E\left\|\Delta_{t-1}(W_{\eff}-I)\right\|^2\\
& + \alpha^2\left(1 + \frac{1}{1 - \lambda_n} \right )^2 \sum_{t=0}^{T-1} \mathbb E\left\|- \gamma \sum_{s=0}^{t}G_sW_{\eff}^{t-s}(W_{\eff}-I)   \right\|^2_F\\
\leq & \alpha^2\left( 2 - \lambda_n \right)^2( 1- \lambda_n) \sum_{t=0}^{T-3}\mathbb E\left\|  \Delta_{t} (W_{\eff}-I)  \right\|^2_F+ \alpha^2(2 - \lambda_n)^3\sum_{t=0}^{T-1}\mathbb E\left\|\Delta_{t-1}(W_{\eff}-I)\right\|^2\\
& + \alpha^2 \left(1 + \frac{1}{1 - \lambda_n} \right )^2 \sum_{t=0}^{T-1} \mathbb E\left\|- \gamma \sum_{s=0}^{t}G_sW_{\eff}^{t-s}(W_{\eff}-I)   \right\|^2_F\\
\leq & \alpha^2(2 -\lambda_n )^2(3 - 2\lambda_n )\sum_{t=0}^{T-1}\mathbb E\left\|\Delta_{t}(W_{\eff}-I)\right\|^2 + \frac{\alpha^2\lambda_n^2\gamma^2}{(1-\lambda_n)^2}\sum_{t=0}^{T-1} \mathbb E\left\|\sum_{s=0}^{t}G_sW_{\eff}^{t-s}(W_{\eff}-I)   \right\|^2_F\\
\leq & \alpha^2(2 -\lambda_n )^2(3 - 2\lambda_n )\sum_{t=0}^{T-1}\mathbb E\left\|\Delta_{t}(W_{\eff}-I)\right\|^2 + \frac{\alpha^2\lambda_n^2\gamma^2}{(1-\lambda_n)^2}\sum_{t=0}^{T-1} \mathbb E\left\|G_t \right\|^2_F. \quad\text{(by Lemma~\ref{lemma:bound_pow2} )}
\end{align*}
the equation above can also be written as
\begin{align*}
\left( 1 - \alpha^2(2 -\lambda_n )^2(3 - 2\lambda_n ) \right)\sum_{t=0}^{T-1}\mathbb E\left\|\Delta_{t}(W_{\eff}-I)\right\|^2 \leq \frac{\alpha^2\lambda_n^2\gamma^2}{(1-\lambda_n)^2}\sum_{t=0}^{T-1} \mathbb E\left\|G_t \right\|^2_F.
\end{align*}
So if 
\begin{align*}
&1 - \alpha^2(2 - \lambda_n )^2(3 - 2\lambda_n) \geq 0\\
&(2 - \lambda_n )^2(3 - 2\lambda_n) \leq \frac{1}{\alpha^2}
\end{align*}
then we have
\begin{align*}
\sum_{t=0}^{T-1}\mathbb E\left\|\Delta_{t}(W_{\eff}-I)\right\|^2 \leq C_4\sum_{t=0}^{T-1} \mathbb E\left\|G_t \right\|^2_F.
\end{align*}

\end{proof}

\begin{lemma} \label{lemma:bound_G}
Under the Assumption \ref{ass:global}, when using Algorithm~\ref{alg:de_ec}, we have
\begin{align*}
\mathbb{E}\left\|G(X_t,\xi_t)\right\|^2_F\leq &  n\sigma^2+3L^2\mathbb{E}\left\|X_t(I - A_n)\right\|^2_F+3n\zeta^2+3n\mathbb{E}\left\|\nabla f(\overline{\bm{x}}_t)\right\|^2.
\end{align*}
\end{lemma}

\begin{proof}
Notice that
\begin{align*}
\mathbb{E}\left\|G(X_t,\xi_t)\right\|^2_F
= \sum_{i=1}^{n}\mathbb{E}\left\|\nabla F_i(x_t^{(i)};\xi_t^{(i)})\right\|^2.
\end{align*}
We next estimate the upper bound of $\mathbb{E}\left\|\nabla F_i(\bm{x}_t^{(i)};\xi_t^{(i)})\right\|^2$ in the following
\begin{align*}
\mathbb{E}\left\|\nabla F_i(\bm{x}_t^{(i)};\xi_t^{(i)})\right\|^2
= & \mathbb{E}\left\|\left(\nabla F_i(\bm{x}_t^{(i)};\xi_t^{(i)})-\nabla f_i(\bm{x}_t^{(i)})\right)+\nabla f_i(\bm{x}_t^{(i)})\right\|^2\\
= & \mathbb{E}\left\|\nabla F_i(\bm{x}_t^{(i)};\xi_t^{(i)})-\nabla f_i(\bm{x}_t^{(i)})\right\|^2+\mathbb{E}\left\|\nabla f_i(\bm{x}_t^{(i)})\right\|^2\\
 & + 2\mathbb{E}\left\langle\mathbb{E}_{\xi_t}\nabla F_i(\bm{x}_t^{(i)};\xi_t^{(i)})-\nabla f_i(\bm{x}_t^{(i)}),\nabla f_i(\bm{x}_t^{(i)})\right\rangle\\
= &\mathbb{E}\left\|\nabla F_i(\bm{x}_t^{(i)};\xi_t^{(i)})-\nabla f_i(\bm{x}_t^{(i)})\right\|^2+\mathbb{E}\left\|\nabla f_i(\bm{x}_t^{(i)})\right\|^2\\
\leq & \sigma^2 + \mathbb{E}\left\|\left(\nabla f_i(\bm{x}_t^{(i)})-\nabla f_i(\overline{\bm{x}}_t)\right)+\left(\nabla f_i(\overline{\bm{x}}_t)-\nabla f(\overline{\bm{x}}_t)\right)+\nabla f(\overline{\bm{x}}_t)\right\|^2\\
\leq & \sigma^2 + 3\mathbb{E}\left\|\nabla f_i(\bm{x}_t^{(i)})-\nabla f_i(\overline{\bm{x}}_t)\right\|^2 + 3\mathbb{E}\left\|\nabla f_i(\overline{\bm{x}}_t)-\nabla f(\overline{\bm{x}}_t)\right\|^2\\
& + 3\mathbb{E}\left\|\nabla f(\overline{\bm{x}}_t)\right\|^2\\
\leq & \sigma^2+3L^2\mathbb{E}\left\|\overline{\bm{x}}_t-\bm{x}_t^{(i)}\right\|^2+3\zeta^2+3\mathbb{E}\left\|\nabla f(\overline{\bm{x}}_t)\right\|^2,
\end{align*}
which means
\begin{align*}
\mathbb{E}\left\|G(X_t,\xi_t)\right\|^2\leq &
\sum_{i=1}^{n}\left\|\nabla F_i(\bm{x}_t^{(i)};\xi_t^{(i)})\right\|^2\\
\leq & n\sigma^2+3L^2\sum_{i=1}^n\mathbb{E}\left\|\overline{\bm{x}}_t-\bm{x}_t^{(i)}\right\|^2+3n\zeta^2+3n\mathbb{E}\left\|\nabla f(\overline{\bm{x}}_t)\right\|^2\\
 = & n\sigma^2+3L^2\mathbb{E}\left\|X_t(I - A_n)\right\|^2_F+3n\zeta^2+3n\mathbb{E}\left\|\nabla f(\overline{\bm{x}}_t)\right\|^2.
\end{align*}

\end{proof}

\begin{lemma}\label{lemma:bound_pow0}
For any matrix sequence $\{M_t\}$, we have
\begin{align*}
\sum_{t=0}^{T-1}\left\| \sum_{s=0}^t M_s(I - A_n)W_{\eff}^{t-s} \right\|^2_F \leq  \frac{1}{(1-\lambda_2)^2}\sum_{t=0}^{T-1}\left\|  M_t \right\|^2_F.
\end{align*}
\end{lemma}

\begin{proof}
Since $W_{\eff}$ is symmetric, so we decompose it as
\begin{align*}
W_{\eff} = P\Lambda P^{\top} = \sum_{i=1}^n \lambda_i(W_{\eff})\bm{v}^{(i)}\left(\bm{v}^{(i)}\right)^{\top},
\end{align*}
where $\bm{v}^{(i)}$ is the corresponding eigenvector of $\lambda_i(W_{\eff})$. Meanwhile, we have 
\begin{align*}
W\frac{\bm{1}}{n} = \frac{\bm{1}}{n},
\end{align*}
 which means
 \begin{align*}
 W_{\eff} = A_n + \sum_{i=2}^n \lambda_i(W_{\eff})\bm{v}^{(i)}\left(\bm{v}^{(i)}\right)^{\top}.
 \end{align*}

Since $W_{\eff}^t = P\Lambda^t P^{\top}$, we have for any $r$ and $s$ 
\begin{align}
\left\|  M_r W_{\eff}^s-M_sA_n \right\|^2_F = &\left\|  M_r (W^s_{\eff}-A_n) \right\|^2_F \nonumber\\
= & \left\| M_r P\Lambda^sP^{\top}-M_r P\begin{pmatrix}
1,0,\cdots,0\\0,0,\cdots,0\\ \cdots \\0,0,\cdots,0
\end{pmatrix}P^{\top}\right\|^2_F \nonumber \\
= & \left\| M_r P\Lambda^s-M_r P\begin{pmatrix}
1,0,\cdots,0\\0,0,\cdots,0\\ \cdots \\0,0,\cdots,0
\end{pmatrix}\right\|^2_F \nonumber \\
= & \left\| M_r P\begin{pmatrix}
&0,&0,&0,&\cdots,&0\\
&0,&\lambda_2^s,&0,&\cdots,&0\\
&0,&0,&\lambda_3^s,&\cdots,&0\\
& \hdotsfor{5}\\
&0,&0,&0,&\cdots,&\lambda_n^s
\end{pmatrix}\right\|^2_F \nonumber\\
\leq & \left\|\lambda_2^{t}M_r P\right\|^2_F \nonumber\\
= & \left\|\lambda_2^{s}M_r \right\|^2_F. \label{eq:lemma6_1}
\end{align}

Therefore, we can have
\begin{align*}
\sum_{t=0}^{T-1} \left\| \sum_{s=0}^t M_s(I - A_n)W_{\eff}^{t-s} \right\|^2_F = & \sum_{t=0}^{T-1}\left\| \sum_{s=0}^t M_s W_{\eff}^{t-s} -M_sA_n \right\|^2_F   \\
\leq & \sum_{t=0}^{T-1}  \left ( \sum_{s=0}^t  \left\| M_s W_{\eff}^{t-s} -M_sA_n \right\|_F \right ) ^2 \\
\leq & \sum_{t=0}^{T-1}  \left ( \sum_{s=0}^t  \lambda_2^{t-s} \left \| M_s\right\|_F \right ) ^2, \quad \text{ (due to \eqref{eq:lemma6_1})} \\
\leq & \frac{1}{(1-\lambda_2)^2}\sum_{t=0}^{T-1}  \left \| M_t \right\|_F ^2, \quad \text{ (due to Lemma \ref{lemma:seq})}
\end{align*}
where the first equality holds due to $A_n W_\eff = A_n [(1-\eta)I + \eta  W] = (1-\eta ) A_n + \eta A_n = A_n$.

\end{proof}

\begin{lemma}\label{lemma:bound_diff_X}
For $\OA$, if 
\begin{align*}
(2 - \lambda_n )^2(3 - 2\lambda_n) \leq & \frac{1}{\alpha^2},\\
 1 -3C_5L^2\gamma^2 \geq & 0,
\end{align*}
we have
\begin{align*}
\sum_{t=0}^{T-1} \mathbb E\left\|X_t(I - A_n)\right\|_F^2 \leq \frac{C_5\gamma^2n\sigma^2 T}{1 -3C_5L^2\gamma^2} + \frac{3nC_5\gamma^2\zeta^2 T}{1 -3C_5L^2\gamma^2} + \sum_{t=0}^{T-1}\frac{3nC_5 \gamma^2\mathbb{E}\left\|\nabla f(\overline{\bm{x}}_t)\right\|^2}{1 -3C_5L^2\gamma^2}.
\end{align*}

\end{lemma}

\begin{proof}
From \eqref{supp:key_eq1}, we have
\begin{align*}
X_t(I - A_n) =& - \gamma \sum_{s=0}^{t-1}G_sW_{\eff}^{t-s}(I - A_n) + \sum_{s=0}^{t-2} \Delta_{s} (W_{\eff}-I)^2W_{\eff}^{t-2-s}(I - A_n) \\
&- \Delta_{t-1}(W_{\eff}-I)(I - A_n)\\
= & - \gamma \sum_{s=0}^{t-1}G_sW_{\eff}^{t-s}(I - A_n) + \sum_{s=0}^{t-2} \Delta_{s} (W_{\eff}-I)^2W_{\eff}^{t-2-s} - \Delta_{t-1}(W_{\eff}-I),
\end{align*}
so we have
\begin{align*}
\left\|X_t(I - A_n)\right\|_F^2 \leq & 3 \gamma^2 \left\|\sum_{s=0}^{t-1}G_sW_{\eff}^{t-s}(I - A_n)\right\|^2_F + 3\left\| \sum_{s=0}^{t-2} \Delta_{s} (W_{\eff}-I)^2W_{\eff}^{t-2-s} \right\|^2_F \\
&+ 3\left\| \Delta_{t-1}(W_{\eff}-I) \right\|_F^2,
\end{align*}
which leads to 
\begin{align*}
&\sum_{t=0}^{T-1} \mathbb E\left\|X_t(I - A_n)\right\|_F^2\\
 \leq & 3 \gamma^2\sum_{t=0}^{T-1} \mathbb E \left\|\sum_{s=0}^{t-1}G_sW_{\eff}^{t-s}(I - A_n)\right\|^2_F + 3\sum_{t=0}^{T-1} \mathbb E\left\| \sum_{s=0}^{t-2} \Delta_{s} (W_{\eff}-I)^2W_{\eff}^{t-2-s} \right\|^2_F\\
 & + 3\sum_{t=0}^{T-1} \mathbb E\left\| \Delta_{t-1}(W_{\eff}-I) \right\|_F^2.\numberthis\label{lemma:bound_diff_eq1}
\end{align*}
From Lemma(Too be added) we have
\begin{align*}
\sum_{t=0}^{T-1}\left\|\sum_{s=0}^{t-1}G_sW_{\eff}^{t-s}(I - A_n)\right\|^2_F \leq \frac{1}{(1-\lambda_2)^2}\sum_{t=0}^{T-2}\left\|G_t\right\|^2_F,\numberthis\label{lemma:bound_diff_eq1_2}
\end{align*}
using Lemma~\ref{lemma:bound_pow2}, we have
\begin{align*}
\sum_{t=0}^{T-1}\left\| \sum_{s=0}^{t-2} \Delta_{s} (W_{\eff}-I)^2W_{\eff}^{t-2-s} \right\|^2_F \leq \sum_{t=0}^{T-3} \left\|\Delta_{t}(W_{\eff}-I) \right\|_F^2.\numberthis\label{lemma:bound_diff_eq1_3}
\end{align*}
Taking \eqref{lemma:bound_diff_eq1_2} and \eqref{lemma:bound_diff_eq1_3} into \eqref{lemma:bound_diff_eq1} we get
\begin{align*}
&\sum_{t=0}^{T-1} \mathbb E\left\|X_t(I - A_n)\right\|_F^2\\
\leq & \frac{3\gamma^2}{(1-\lambda_2)^2}\sum_{t=0}^{T-2}\left\|G_t\right\|^2_F + 3\sum_{t=0}^{T-3} \left\|\Delta_{t}(W_{\eff}-I) \right\|_F^2 + 3\sum_{t=0}^{T-1} \mathbb E\left\| \Delta_{t-1}(W_{\eff}-I) \right\|_F^2\\
\leq & \frac{3\gamma^2}{(1-\lambda_2)^2}\sum_{t=0}^{T-2}\left\|G_t\right\|^2_F + 6\sum_{t=0}^{T-1} \left\|\Delta_{t}(W_{\eff}-I) \right\|_F^2.
\end{align*}
Using Lemma~\ref{lemma:bound_compress_error} to bound $\sum_{t=0}^{T-1} \left\|\Delta_{t}(W_{\eff}-I) \right\|_F^2$, then we get
\begin{align*}
\sum_{t=0}^{T-1} \mathbb E\left\|X_t(I - A_n)\right\|_F^2 \leq &\frac{3\gamma^2}{(1-\lambda_2)^2}\sum_{t=0}^{T-2}\left\|G_t\right\|^2_F + 6C_4\gamma^2\sum_{t=0}^{T-1} \left\|G_t\right\|^2_F\\
\leq & \left( \frac{3}{(1-\lambda_2)^2} + 6 C_4\right)\gamma^2\sum_{t=0}^{T-1} \left\|G_t\right\|^2_F.
\end{align*}
Using Lemma~\ref{lemma:bound_G}, the equation above becomes
\begin{align*}
& \sum_{t=0}^{T-1} \mathbb E\left\|X_t(I - A_n)\right\|_F^2 \\
\leq &\left( \frac{3\gamma^2}{(1-\lambda_2)^2} + 6 C_4\gamma^2\right) \left( n\sigma^2 T +3L^2\mathbb{E}\left\|X_t(I - A_n)\right\|^2_F T +3n\zeta^2T+3n\sum_{t=0}^{T-1}\mathbb{E}\left\|\nabla f(\overline{\bm{x}}_t)\right\|^2  \right)\\
\leq & C_5 \gamma^2 n\sigma^2T + 3C_5 \gamma^2L^2\mathbb{E}\left\|X_t(I - A_n)\right\|^2_F T+3nC_5 \gamma^2\zeta^2 T+3nC_5 \gamma^2\sum_{t=0}^{T-1}\mathbb{E}\left\|\nabla f(\overline{\bm{x}}_t)\right\|^2,
\end{align*}
which leads to 
\begin{align*}
&\left( 1 - 3C_5L^2\gamma^2  \right)\sum_{t=0}^{T-1} \mathbb E\left\|X_t(I - A_n)\right\|_F^2\\
 \leq &  \frac{C_5\gamma^2n\sigma^2 T}{1 -3C_5L^2\gamma^2} + \frac{3nC_5\gamma^2\zeta^2 T}{1 -3C_5L^2\gamma^2} + \sum_{t=0}^{T-1}\frac{3nC_5 \gamma^2\mathbb{E}\left\|\nabla f(\overline{\bm{x}}_t)\right\|^2}{1 -3C_5L^2\gamma^2}.
\end{align*}
So if $ 1 -3C_5L^2\gamma^2 \geq 0 $, then we have
\begin{align*}
\sum_{t=0}^{T-1} \mathbb E\left\|X_t(I - A_n)\right\|_F^2 \leq \frac{C_5\gamma^2n\sigma^2 T}{1 -3C_5L^2\gamma^2} + \frac{3nC_5\gamma^2\zeta^2 T}{1 -3C_5L^2\gamma^2} + \sum_{t=0}^{T-1}\frac{3nC_5 \gamma^2\mathbb{E}\left\|\nabla f(\overline{\bm{x}}_t)\right\|^2}{1 -3C_5L^2\gamma^2}.
\end{align*}

\end{proof}

\section*{Proof of Theorem~\ref{theo} }
Recall that the updating rule of $\OA$ can be written in an equivalent form as follows
\begin{align*}
\Delta_t =& (X_t - \gamma G_t) + \Delta_{t-1} - C_{\omega}\left[ (X_t - \gamma G_t) + \Delta_{t-1} \right]\\
X_{t+1} =& (X_t - \gamma G_t) W_{\eff}+ \eta (\Delta_{t-1} - \Delta_t) (W_{\eff}-I)
\end{align*}
where 
\begin{align*}
W_{\eff} := (1-\eta)I + \eta W.
\end{align*}
Below we assume that $\eta\leq \frac{1}{2}$, so that $W_{\eff}\succeq 0$.

Our proof of Theorem \ref{theo} is mainly based on Lemma \ref{lemma:key_theo_lemma}. We give the proof of Theorem \ref{theo} as follows.

\paragraph{Proof of Theorem~\ref{theo}}
\begin{proof}
 If $\eta$ satisfies
\begin{align*}
\eta \leq & \min\left\{\frac{1}{2},\frac{\alpha^{-\frac{2}{3}} - 1}{4}\right\},
\end{align*}
and
\begin{align*}
\lambda_n = (1-\eta ) + \eta\lambda_n(W) > 1 - \eta - \eta = 1-2\eta,
\end{align*}
we have
\begin{align*}
(2 - \lambda_n )^2(3 - 2\lambda_n) < (1 + 2\eta)^2(1 + 4\eta)\leq (1 + 4\eta)^3\leq \frac{1}{\alpha^2}.
\end{align*}
So if 
\begin{align*}
(1 + 4\eta)^3 \leq & \frac{1}{\alpha^2},\quad\text{and}\quad
\eta \leq \frac{\alpha^{-\frac{2}{3}} - 1}{4}.
\end{align*}
Then we can ensure that
\begin{align*}
W_{\eff} \succeq & 0,\\
(2 - \lambda_n )^2(3 - 2\lambda_n) \leq & \frac{1}{\alpha^2}.
\end{align*}
{Then Theorem~\ref{theo} can be easily verified by replacing $\lambda_2 = 1-\eta + \eta \lambda_2(W) $ and $\lambda_n= 1-\eta + \eta \lambda_n(W) $ in Lemma~\ref{lemma:key_theo_lemma}.} 
%
%
\end{proof}

\section{Proof of Corollary \ref{coro}}
\begin{proof}
From Theorem\ref{theo}, we have
\begin{align*}
&\left(\frac{\gamma}{2} -\frac{3C_2L^2 \gamma^3}{2 -6C_2L^2\gamma^2} \right)\sum_{t=0}^T E \left\|\nabla f(\overline{\bm{x}}_t)\right\|^2\\
\leq & \mathbb E f\left(\overline{\bm{x}}_0 \right) - \mathbb E f\left(\bm{x}^* \right)  + \left(\frac{L\gamma^2}{2n} + \frac{C_2L^2\gamma^3}{2 -6C_2L^2\gamma^2}\right)\sigma^2T + \frac{3C_2L^2\gamma^3\zeta^2T}{2 -6C_2L^2\gamma^2},
\end{align*}
which can be rewritten as
\begin{align*}
& \left(1 -\frac{3C_2L^2 \gamma^2}{1 -3C_2L^2\gamma^2} \right)\sum_{t=0}^T E \left\|\nabla f(\overline{\bm{x}}_t)\right\|^2\\
\leq & \frac{2 \mathbb E f\left(\overline{\bm{x}}_0 \right) - 2\mathbb E f\left(\bm{x}^* \right) }{\gamma}   + \left(\frac{L\gamma}{n} + \frac{C_2L^2\gamma^2}{1 -3C_2L^2\gamma^2}\right)\sigma^2T + \frac{3C_2L^2\gamma^2\zeta^2T}{1 -3C_2L^2\gamma^2}.
\end{align*}
So if 
\begin{align*}
\gamma \leq & \frac{1}{3L\sqrt{C_2}},\\
3C_2L^2\gamma^2 \leq & \frac{1}{3},
\end{align*}
we have
\begin{align*}
 1- \frac{3C_2L^2 \gamma^2}{1 -3C_2L^2\gamma^2} \leq 1 - \frac{1}{2} = \frac{1}{2}.
\end{align*}
Choosing $\gamma = \frac{1}{3L\sqrt{C_2} + \sigma\sqrt{\frac{T}{n}} + \zeta^{\frac{2}{3}}T^{\frac{1}{3}}} $, it can be easily verified that
\begin{align*}
\frac{1}{T}\sum_{t=0}^T E \left\|\nabla f(\overline{\bm{x}}_t)\right\|^2 \lesssim \left(\frac{1}{\sqrt{nT}} + \frac{C_2}{T}\right)\sigma +  
\frac{C_2\zeta^{\frac{2}{3}}}{T^{\frac{1}{3}}} + \frac{1 + \sqrt{C_2} }{T}.
\end{align*}

Since $\eta = \min\{\frac{1}{2},\frac{\alpha^{\frac{-2}{3}} - 1}{4}\} $, so we have
\begin{align*}
C_0 = & \eta(1 - \lambda_n(W))\leq  2\eta\\
C_1 = & \frac{\alpha^2}{\left(1 - \alpha^2\left(1 + C_0\right)^2(1 + 2C_0) \right)C_0^2} \leq \frac{\alpha^2}{4\left(1 - \alpha^2\left(1 + 4\eta\right)^3 \right)\eta^2}\leq \frac{\alpha^2}{2\left( \alpha^{\frac{-2}{3}} - 1 \right)},
\end{align*}
which gives us
\begin{align*}
C_2 = & \frac{3}{\eta^2(1-\lambda_2(W))^2} + 6C_1 \lesssim \frac{1}{(1 - \lambda_2(W))^2}\left ( 1+  \frac{\alpha^2}{2\left( \alpha^{\frac{-2}{3}} - 1 \right)} \right)
\end{align*}

\end{proof}

\end{document}